\documentclass[journal]{IEEEtran}
\usepackage{amsmath,amssymb,amsthm}
\usepackage{graphicx,subfigure,epstopdf}
\usepackage[comma,numbers,square,sort&compress]{natbib}
\usepackage{epstopdf}
\usepackage{epic}
\usepackage{subfigure}
\usepackage{graphicx}
\usepackage{epstopdf}
\usepackage{epsfig}
\usepackage{arydshln}
\usepackage{enumerate}
\usepackage[utf8]{inputenc}
\usepackage[english]{babel}
\usepackage{fancyhdr}
\usepackage{lastpage}
\usepackage{color}
\usepackage{dsfont}
\usepackage{changes}
\usepackage{pgf,tikz,epstopdf}

\usepackage{pgf}
\usepackage{indentfirst}
\usetikzlibrary{positioning}

\newcommand{\EQ}{\begin{eqnarray}}

\newcommand{\EN}{\end{eqnarray}}

\newcommand{\EQQ}{\begin{eqnarray*}}
\newcount\mycount
\newcommand{\ENN}{\end{eqnarray*}}

\newtheorem{thm}{\bf \em{Theorem}}

\newtheorem{lem}{\bf \em{Lemma}}
\newtheorem{rem}{\bf \em{Remark}}
\newtheorem{defi}{Definition}

\newtheorem{prob}{\bf \em{Problem}}

\allowdisplaybreaks
\usepackage{caption}
\begin{document}
\title{Adaptive Consensus and Parameter Estimation of Multi-Agent Systems with
An Uncertain Leader}
\author{Shimin~Wang and Xiangyu Meng \thanks{%
Shimin~Wang is with the Department of Electrical and Computer Engineering
11-203, Donadeo Innovation Centre for Engineering University of Alberta 9211
- 116 Street NW, Edmonton, Alberta, Canada. E-mail: shimin1@ualberta.ca.}
\thanks{%
Xiangyu Meng is with the Division of Electrical and Computer Engineering,
Louisiana State University, Baton Rouge, LA 70803, USA. E-mail:
xmeng5@lsu.edu. } }
\maketitle

\begin{abstract}
In this note, the problem of simultaneous leader-following consensus and parameter estimation is
studied for a class of multi-agent systems subject to an uncertain
leader system. The leader system is described by a sum of
sinusoids with unknown amplitudes, frequencies and phases. A
distributed adaptive observer is established for each agent to
estimate the unknown frequencies of the leader.
It is shown that if the signal of the leader is sufficiently
rich, the estimation errors of the unknown frequencies converge
to zero asymptotically for all the agents. Based on the designed distributed adaptive observer, a distributed adaptive control law is
synthesized for each agent to solve the leader-following
consensus problem.
\end{abstract}

\begin{IEEEkeywords}Consensus, adaptive observer, multi-agent systems, distributed control, parameter estimation.\end{IEEEkeywords}

\section{Introduction}

\IEEEPARstart{M}{ulti-agent} coordination has attracted considerable attention, for example, see \cite{jadbabaie2003coordination,
olfati2007consensus,liu2016finite,cao2010decentralized,cao2011formation,wangsong2017,feng2019finite,lin2018distributed}, where a significant amount of research has been directed towards the leader-following consensus problem \cite{cai2016leaderEL,xiao2017adaptive,chen2019feedforward}. Knowing the parameters of the
leader plays a key role in the design of distributed control laws for solving the
leader-following consensus problem of multi-agent systems. The leader-following consensus problem and the leader-following flocking problem for a leader with known parameters have been solved by \cite{su2011cooperative,cai2016leaderEL} and \cite{yu2010distributed}, respectively, where
the system matrix of the
leader is used to design a distributed control law.

Without the knowledge of the leader's parameters, adaptive observers were utilized in \cite{modares2016optimal,wu2017adaptive,xiao2017adaptive} to estimate the states of the leader. Since only a subset of the followers is able to access the signal of the leader, all followers take advantage of the connectivity of the topological graph to share information so as to achieve leader-following consensus. The aforementioned references \cite{modares2016optimal, wu2017adaptive, xiao2017adaptive} designed adaptive observers by sharing the estimated states of the leader. But the states of the leader may be inaccessible. This constraint has been relaxed in \cite{lu2019leader} by using the output of the leader and all followers exchange their estimated output of the leader. We also notice that another type of observers based on sliding mode have been proposed, for example, see \cite{cao2010decentralized, zhao2015distributed}, where bounded velocity and bounded acceleration are assumed in \cite{cao2010decentralized} and \cite{zhao2015distributed}, respectively. While the sliding mode estimator methods are able to show finite-time estimation, they suffer from chattering \cite{chung2009cooperative}, which remains as a mathematical challenge. Even though the leader-following consensus problem with an uncertain leader has been solved by showing that the outputs of all followers converge to the leader's signal, the convergence of the estimation errors of leader's unknown parameters is not analyzed in the above references.

The objective and primary contribution of this note is to propose a formal framework for simultaneous leader-following consensus and parameter estimation of an uncertain leader in multi-agent networks. This is, however, a challenging task. Motivated by traditional adaptive observers \cite{rimon1992new, narendra1974stableI, narendra1974stableII, ioannou1996robust, adetola2014adaptive, hsuliu1999, jiang2019parameter,li2019kernel,carnevale2013semi,zhu2015parameter} for single agent systems, we make an attempt by extending them to multi-agent systems. Standard approaches for single agent systems do not apply straightforwardly to the multi-agent setting due to the network constraint that the leader is indistinguishable from followers.  Analysis of the parameter estimation in the leader-following consensus problem with an uncertain leader is relatively less well understood, for example, see \cite{wanghuang2018}. We distinguish our work with \cite{wanghuang2018} by designing adaptive observes with output information. In this note, the leader is described by the sum of sinusoids with unknown amplitudes, frequencies and phases. The signal of the leader can be regarded as the output of a linear system with unknown parameters. Assuming the upper bound of these frequencies are known, we can design an output based distributed adaptive observer for each agent to estimate the virtual states of the leader. Furthermore, if the signal of the leader is sufficiently rich, then the estimated parameters asymptotically converge to the actual parameters of the leader. Based on the estimated virtual states of the leader, a local controller is designed for each agent to track the signal of the leader.

In summary, this note makes several contributions toward leader-following consensus and parameter estimation of multi-agent systems with an uncertain leader. First, distributed adaptive observers are designed to estimate both the virtual states and unknown parameters of the leader without knowing the amplitudes and derivative of the leader's signal. In this sense, the sliding mode based approaches are inapplicable. Second, a sufficient condition is identified to guarantee that the parameter estimation errors converge to zero asymptotically. Such analysis is missing in most existing references. Third, agents only need to share the estimated output information, which includes the distributed full state adaptive observers as a special case.  Last, the proposed framework can be easily extended to other cooperative problems, such as cooperative output regulation of linear multi-agent systems, leader-following consensus of multiple Euler-Lagrange systems and attitude synchronization of multiple rigid body systems.

The rest of this note is organized as follows. In Section \ref{section1},
we formulate the problem of simultaneous leader-following consensus and parameter estimation (SLFCAPE). Section \ref%
{section2} is devoted to the design of distributed adaptive observers and distributed observer based controllers. In Section \ref{section4}, detailed analyses are given to show the effectiveness of the proposed control strategy for solving the SLFCAPE problem with an uncertain leader system. A simulation example is given in Section \ref%
{section5}. Finally, we conclude this note in Section~\ref{section6}.

\section{Problem Formulation}

\label{section1}

\subsection{Multi-agent Network}
\label{section1a}

\label{ss1} As in \cite{cai2016leaderEL}, the multi-agent system is composed
of a leader and $N$ followers. The network topology of the multi-agent
system is described by a graph $\bar{\mathcal{G}}=\left(\bar{\mathcal{V}},%
\bar{\mathcal{E}}\right)$ with $\bar{\mathcal{V}}=\{1,\dots,N,N+1\}$ and $\bar{%
\mathcal{E}}\subseteq\left[\bar{\mathcal{V}}\right]^2$, which is the
2-element subsets of $\bar{\mathcal{V}}$. Here node $N+1$ is associated with
the leader and node $i$ is associated with follower $i$ for $i=1,\dots,N$.
For $i=1,\dots,N,N+1$, $j=1,\dots,N$, $(i,j) \in \bar{\mathcal{E}} $ if and
only if agent $j$ can receive information from agent $i$. Let $\bar{\mathcal{%
N}}_i=\{j|(j,i)\in \bar{\mathcal{E}}\}$ denote the neighborhood set of agent
$i$. Let $\mathcal{G}=(\mathcal{V},\mathcal{E})$ denote the induced subgraph
of $\bar{\mathcal{G}}$ with $\mathcal{V}=\{1,\dots,N\}$. Assume that $%
\mathcal{\bar{G}}$ contains a spanning arborescence with node $N+1$ as the root
and $\mathcal{G}$ is an undirected graph. Let $\mathcal{\bar{L}}$ be the
Laplacian matrix of the graph $\mathcal{\bar{G}}$, and $H$ is obtained by
deleting the last row and column of $\mathcal{\bar{L}}$. Then, $H$ is a
positive definite symmetric matrix with $\lambda_1 >0$ being its smallest eigenvalue \cite{su2011cooperative}. More details of the graph theory can be found in \cite%
{godsil2013algebraic}.

\subsection{Leader Dynamics}
The signal of the leader system is described by:
\begin{align}  \label{leaderall}
y_{N+1}(t)=\sum_{k=1}^l\varphi_k\sin\left(\omega_k t +\psi_k\right),
\end{align}
where $\varphi_k$, $\psi_k$ and $\omega_k>0$, for $k=1,\ldots,l$ are
unknown amplitudes, phases and frequencies, respectively. Assume that $%
\omega_k<\bar{\omega}$ for $k=1,\ldots,l$, where $\bar{\omega}$ is a known
upper bound.

Also assume that the signal $y_{N+1}$ is sufficiently rich of order
$2l$, that is, it consists of at least $l$ distinct frequencies \cite{ioannou1996robust}.

\begin{rem}The leader-following problem for an unknown leader with bounded velocity and bounded acceleration are discussed in \cite{cao2010decentralized} and \cite{zhao2015distributed}, respectively, where the sliding mode estimators are used to handle the uncertainty of the leader. However, such approaches are inapplicable here since $\varphi_k$ is unknown for each follower, for $k=1,\ldots,l$.%, the approaches based on sliding
\end{rem}
\subsection{Follower Dynamics}
The dynamics of follower $i$ are described by the following single-input and
single-output system:
\begin{subequations}\label{MARINEVESSEL1}
\begin{align}
\dot{x}_{i,s}&=x_{i,s+1},\quad s=1,\ldots,r-1, \nonumber\\
\dot{x}_{i,r}&=u_i,\\
y_i&=x_{i,1},
\end{align}
\end{subequations}
 where $x_i=\mathrm{col}\left(x_{i,1},\ldots,x_{i,r}%
\right)\in \mathds{R}^{r}$, $u_i\in \mathds{R}$, and $y_i\in \mathds{R}$ are the state vector, control input and the output of follower~$i$, respectively, for $i=1,\ldots,N$.

\subsection{Objective}
The SLFCAPE problem considered in this paper is formulated as
follows.
\begin{prob}[SLFCAPE Problem]
\label{ldlesp} Given a multi-agent network $\mathcal{\bar{G}}$ with the
leader dynamics (\ref{leaderall}) and the follower dynamics (\ref%
{MARINEVESSEL1}), find a distributed parameter estimator and a distributed controller such that $%
x_i(t)$ is bounded for all $t\geq0$ and
\begin{equation*}
\lim\limits_{t\rightarrow\infty}\left(y_i\left(t\right)-y_{N+1}\left(t\right)%
\right)=0~~\text{and}~~\lim\limits_{t\rightarrow\infty}\left(\hat{\omega}_{i,k}\left(t\right)-\omega_k%
\right)=0,
\end{equation*}
where $\hat{\omega}_{i,k}(t)$ is an estimate of $\omega_k$ for $k=1,\ldots,l$, for any initial conditions $x_i(0)$, $i=1,\ldots,N$.
\end{prob}

\section{Distributed Adaptive Control Design}\label{section2}
The signal $y_{N+1}\in \mathds{R}$ can be regarded as the output of the following virtual linear system,%
\begin{subequations}\label{leader2}
\begin{align}
\dot{v}&=g(\theta)v=M v-\sum\limits_{k=1}^{l}\theta_k E_{2k}y_{N+1} \\
y_{N+1}&=C v
\end{align}
\end{subequations}
where $E_{2k}=\mathrm{col}\left(0_{1 \times (2k-1)},1, 0_{1 \times
(2l-2k)}\right)$, $v\in \mathds{R}^{2l}$,
\begin{align}  \label{MC}
C^T=\mathrm{col}\left(1,0_{(2l-1)\times 1}\right)\in \mathds{R}^{2l},~~
M=\left[
\begin{array}{cc}
0 & I_{2l-1} \\
0 & 0
\end{array}
\right],\end{align}
 and the matrix function $%
g(\cdot):\mathds{R}^l\mapsto \mathds{R}^{ 2l \times 2l}$ defined as
\begin{align}  \label{skewopen}
g(\theta)=\left[
            \begin{array}{ccccccc}
              0 & 1 & 0 & \cdots & 0 & 0 & 0 \\
              -\theta_1 & 0 & 1 &\cdots & 0 & 0 & 0 \\
              \vdots & \vdots & \ddots & \ddots & \vdots & \vdots & \vdots \\
              0 & 0 & 0 & \ddots & 1 & 0 & 0 \\
              -\theta_{l-1} & 0 & 0 & \cdots & 0 & 1 & 0 \\
              0& 0 & 0& \cdots& 0 & 0 &1 \\
              -\theta_{l}  & 0 & 0 & \cdots & 0& 0 &0 \\
            \end{array}
          \right].
\end{align}
Here $\theta=\mathrm{col}\left(\theta_1,\ldots,\theta_l\right)$ is an
invertible reparameterization of the $l$ unknown parameters $\omega_1^2,\ldots,\omega_l^2$ through the following characteristic polynomial of system \eqref{leader2},%
\begin{align}  \label{char1}
\prod_{k=1}^{l}\left(s^2+\right.\omega_k^2&\left.\right)=s^{2l}+\sum_{k=1}^{l}%
\omega_k^2s^{2(l-1)}+\cdots+\prod_{k=1}^{l}\omega_k^2  \notag \\
=&s^{2l}+\theta_1s^{2(l-1)}+\theta_2s^{2(l-3)}+\cdots+\theta_l.
\end{align}
It is easy to see that $\|\theta\|^2\leq \pi$, where%
\begin{equation}\label{pi1}
\pi=\sum\nolimits_{k=1}^{l}{\binom{l}{k}}^2\bar{\omega}^{4k}.%\label{pi_def}
\end{equation}

Without proceeding further, we review the results of adaptive observers for a single agent system in \cite{marino2002global} in order to provide a better understanding of the proposed distributed adaptive observer. The result in \cite{marino2002global} also shows how to estimate
the unknown frequencies $\omega_k$ for $k=1,\ldots, l$.
\subsection{Centralized Adaptive Observer}
Choose an arbitrary vector $a=\mathrm{col}\left(a_{1},a_{2},\ldots , a_{2l-1}\right)
$ such that all the roots of the polynomial equation
$$\lambda^{2l-1}+a_{1}\lambda^{2l-2}+\cdots+a_{2l-2}\lambda+a_{2l-1}=0$$
have negative-real parts.
Following the filtered transformation steps as in~\cite{marino2002global}, system \eqref{leader2} can
be transformed into the following adaptive observer form
\begin{subequations}
\label{neq4}
\begin{align}
\dot{\eta}&= A\eta+B y_{N+1},  \label{neq4b} \\
\dot{\chi}& = A^T\chi+E y_{N+1},  \label{neq4i} \\
\dot{y}_{N+1}&= E^T \eta+a_{1} y_{N+1}+\chi^T F^{T}\theta,
\end{align}
\end{subequations}
where $\chi\in %
\mathds{R}^{2l-1}$ and $\eta\in %
\mathds{R}^{2l-1}$ are filtered transformation vectors,%
\begin{align}\label{EFmatrix} E&=\mathrm{col}\left(1,0_{(2l-2)\times 1}\right)\in \mathds{R}^{2l-1},\notag\\
F&=\mathrm{blkdiag}\left(\left(I_{l-1}\otimes [-1,0]\right),-1\right),
\end{align}
with $\otimes$ denoting the Kronecker product.  and the matrices $A\in \mathds{R}^{(2l-1)\times (2l-1)}
$ and $B\in \mathds{R}^{2l-1}$ are defined as%
\begin{equation}
A=\left[-a\left.
\begin{tabular}{|l}
$I_{2l-2}$ \\
$0_{1\times \left( 2l-2\right) }$
\end{tabular}%
\right. \right] ,\text{ }B=\left[
\begin{array}{c}
a_{2}-a_{1}a_{1} \\
a_{3}-a_{2}a_{1} \\
\vdots  \\
a_{(2l-1)}-a_{2l-2}a_{1} \\
-a_{2l-1}a_{1}
\end{array}
\right] .  \notag
\end{equation}
The adaptive observer for the system \eqref{neq4} proposed in \cite{marino2002global} is given as follows:
\begin{subequations}
\label{adobser}
\begin{align}
\dot{\hat{\theta}}&=\kappa F \chi(y_{N+1}-\hat{y}), \\
\dot{\hat{y}}&= E^T \eta+a_{1} y_{N+1}+\chi^T F^T\hat{\theta}+\mu(y_{N+1}-\hat{y}),
\end{align}
\end{subequations}
where $\hat{\theta}\in \mathds{R}^{l}$ is the estimation of $\theta$ which contains the unknown frequencies of the leader; $\hat{y}\in %
\mathds{R}$ is the estimation of the leader's output $y_{N+1}$, $\mu >\frac{1}{4}\|A\|$ and $\kappa$ is a positive
constant relating to the adaptation gain.

We now give the so-called persistently exciting property of a signal.
\begin{defi}
\cite{sastry2011adaptive} A bounded piecewise continuous function $%
f:[0,+\infty)\mapsto \mathds{R}^n$ is said to be persistently exciting
if there exist positive constants $\epsilon$, $t_0$, $T_0$ such that,
\begin{equation*}
\frac{1}{T_0}\int^{t+T_0}_{t} f(\tau) f^T (\tau) d\tau\geq\epsilon I_n,~~~~\forall
t\geq t_0
\end{equation*}
\end{defi}
According to \cite{marino2002global}, if $F\chi(t)$ satisfies the condition of persistent excitation with $\chi(t)$ being generated by \eqref{neq4i}, then the states $\hat{\theta}(t)-\theta$ and $\hat{y}(t)-y_{N+1}(t)$ are bounded and tend to zero as $t$
goes to infinity for any initial condition.
\begin{rem}
The adaptive observer in \eqref{adobser} relies on not only $y_{N+1}$ but also $\eta$ and $\chi$.
The signals $\eta$ and $\chi$ in~\eqref{adobser} are generated by (\ref{neq4b}) and (\ref{neq4i}) with $y_{N+1}$ as the input,
respectively.
In the multi-agent setting, some followers are not directly connected to the leader, and the followers that are directly connected to the leader do not know which of its neighbors is the leader.
Thus, the adaptive observer \eqref{adobser} can not be extended to the multi-agent case directly. Such an extension, however, is challenging.
\end{rem}

The design schematic is shown in Fig.~\ref{figschematic}, where the distributed adaptive observer and the observer based controller will be introduced below.
\begin{figure}
\centering
\includegraphics[trim=180 530 -80 115,clip,width=\textwidth]{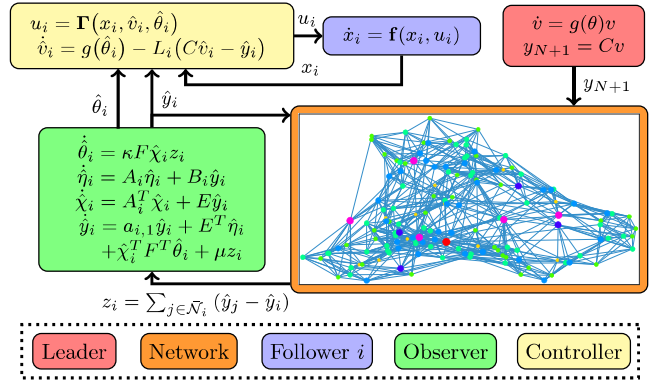}
\caption{Schematic of distributed adaptive observer based controller}\label{figschematic}
\end{figure}
\subsection{Distributed Adaptive Observer}
We are ready to introduce the distributed adaptive
observer for follower $i$:
\begin{subequations}\label{compensator2}
\begin{align}
\dot{\hat{\eta}}_i&= A_i \hat{\eta}_i+B_i \hat{y}_i  \label{compensator2b} \\
\dot{\hat{\chi}}_i&=A_i^T\hat{\chi}_i+E \hat{y}_i  \label{compensator2d} \\
\dot{\hat{\theta}}_i&=\kappa F \hat{\chi}_{i}\sum_{j \in \mathcal{\bar{N}}%
_i}(\hat{y}_j-\hat{y}_i)  \label{compensator2c} \\
\dot{\hat{y}}_i&=E^{T}\hat{\eta}_{i}+a_{i,1}\hat{y}_i+\hat{\chi}_{i}^T F^{T}\hat{\theta}_i+\mu\sum_{j \in \mathcal{\bar{N}}_i}(\hat{y}_j-\hat{y}_i)
\label{compensator2a}
\end{align}
\end{subequations}
where
$a_i=\mathrm{col}\left(a_{i,1},a_{i,2},\ldots , a_{i,(2l-1)}\right)$ is chosen such that all the roots of the polynomial equation
\begin{align}  \label{aplychar1}
\lambda^{2l-1}+a_{i,1}\lambda^{2l-2}+\cdots+a_{i,(2l-2)}\lambda+a_{i,(2l-1)}=0,
\end{align}
have negative-real parts, $E$ and $F$ are defined in \eqref{EFmatrix},
\begin{align}
A_i=\left[-a_i\left.
\begin{tabular}{|l}
$I_{2l-2}$ \\
$0_{1\times \left( 2l-2\right) }$
\end{tabular}\right. \right],
 B_i=\left[
\begin{array}{c}
a_{i,2}-a_{i,1}a_{i,1} \\
a_{i,3}-a_{i,2} a_{i,1} \\
\vdots \\
a_{i,(2l-1)}-a_{i,(2l-2)}a_{i,1} \\
-a_{i,(2l-1)}a_{i,1} \\
\end{array}
\right],\notag
\end{align}
$\hat{\theta}_i\in \mathds{R}^{l}$ is the estimation of $\theta$ which contains the unknown frequencies of the leader; $\hat{y}_i\in %
\mathds{R}$ is the estimation of the leader's output $y_{N+1}$ with $\hat{y}_{N+1}=y_{N+1}$; $\hat{\eta}_i\in \mathds{R}^{2l-1}$ and $\hat{\chi}%
_i\in \mathds{R}^{2l-1}$, for $i=1,\dots,N$. The design parameter $\mu$ is chosen such that % are the estimation of $\chi_i$ and $\eta_i$
\begin{align}  \label{muiequa}
\mu > \frac{2\bar{a}\lambda_{1}+(1+\pi)\lambda _{1}^{2}+\bar{\gamma}_1^2+\bar{\gamma}_2^2}{%
2\lambda_{1}^2},\end{align}
where $\lambda_1$ is the smallest eigenvalue of $H$ defined in Section \ref{section1a}, % be the largest and smallest eigenvalues of $H$,
 $\bar{a}=\max\{a_{1,1},\ldots,a_{N,1}\}$, $\pi$ is defined in (\ref{pi1}), $\bar{\gamma}_1=\max \{\gamma_{1,1},\ldots,\gamma_{N,1}\}$, and $\bar{\gamma}_2=\max \{\gamma_{1,2},\ldots,\gamma_{N,2}\}$ with
\begin{align}\label{gamma12}
\gamma_{i,1}=&\left\|E^{T}(sI-A_i)^{-1}B_i\right\|_\infty\nonumber\\ \gamma_{i,2}=&\left\|F (sI-A_i^T)^{-1}E\right\|_\infty,~~~~i=1,\dots,N.
\end{align}%
%for $$.
\begin{rem}The distributed adaptive observer \eqref{compensator2} is designed based on the adaptive observer form of (\ref{leader2}) inspired by \cite{marino2002global}, where the variables $\hat{\eta}_i$, $\hat{\chi}_i$, $\hat{\theta}_i$ are introduced via the filtered transformation. The adaptive property of the distributed observer refers to the method used by \eqref{compensator2} which adapts to the leader system whose parameters are initially unknown. The observer \eqref{compensator2} relies on only the sum $\sum_{j \in \mathcal{\bar{N}}_i}\hat{y}_j$ which contains all outputs of its neighbours to provide estimations of both the unknown parameter $\theta$ and the signal of the leader $y_{N+1}$.
\end{rem}

\subsection{Distributed Observer Based Controller}
 Define functions $f_{p}(\cdot,\cdot): \mathds{R}^{2l}\times \mathds{R}^{l} \mapsto %
\mathds{R}$ for $p=1,\dots,r+1$ as:
\begin{subequations}
\label{regulatorequa}
\begin{align}
f_1(v,\theta)=&C v \\
f_{s+1}(v,\theta)=&\frac{\partial{f}_s(v,\theta)}{\partial v}g(\theta)v,\quad s=1,\ldots,r,
\end{align}
\end{subequations}
where the matrix function $g(\cdot)$ is defined in \eqref{skewopen}.
Choose $\alpha_1,\dots,\alpha_r$ such that the roots of the
polynomial equation
\begin{align}\label{alphacontrol}\alpha_{r+1}\lambda^{r}+\alpha_r\lambda^{r-1}+\cdots+\alpha_2\lambda+\alpha_1=0\end{align}
have negative-real parts with $\alpha_{r+1}=1$.

The distributed observer based controller for follower $i$ is given as follows:
\begin{subequations}
\label{distrista}
\begin{align}
\dot{\hat{v}}_i& = M \hat{v}_i-L_i(C\hat{v}_i-\hat{y}_i)-\sum_{k=1}^{l}\hat{%
\theta}_{i,k}E_{2k}\hat{y}_i,  \label{distristab} \\
u_i&=\sum_{s=1}^{r+1}\alpha_s f_{s}(\hat{v}_i,\hat{\theta}%
_i)-\sum_{s=1}^{r}\alpha_s x_{i,s},  \label{distristaa}
\end{align}
\end{subequations}
where $M$, $C$ and $E_{2k}$ can be found in \eqref{leader2}, $L_i$ is chosen such that $M-L_iC$ is a Hurwitz matrix, $\hat{v}_i\in \mathds{R}^{2l}$ is the estimation of $v$, $u_i$ is the control input for follower $i$, $\hat{\theta}_{i,k}$ is the $k$th entry of $\hat{\theta}_i$, $\hat{y}_i$ and $\hat{\theta}_i$ are generated by \eqref{compensator2}, for $i=1,\cdots,N$.
\section{Solvability Analysis}\label{section4}
\subsection{Output Estimation Error Analysis}
To analyze the output estimation error of follower $i$, let us transform the dynamics of the leader (\ref{leader2}) into the following adaptive observer form via the filtered transformation proposed in~\cite{marino2002global}:%
\begin{subequations}
\label{neq4a}
\begin{align}
\dot{\eta}_i&= A_i \eta_i+B_i y_{N+1},  \label{neq4ba} \\
\dot{\chi}_i& = A^T_i\chi_i+Ey_{N+1},  \label{neq4ia} \\
\dot{y}_{N+1}&= E^T \eta_i+a_{i,1} y_{N+1}+\chi^T_i F^{T}\theta,
\end{align}
\end{subequations}
where $\chi_i\in \mathds{R}^{2l-1}$ and $\eta_i\in \mathds{R}^{2l-1}$ are filtered transformation vectors, $A_i$, $B_i$, $E$ and $F$ are defined in \eqref{compensator2}. %$\eta_i$ is related to the state $v$.  The variable $\chi_i$ is introduced during the filtered transformation.

Let $\tilde{y}_i=\hat{y}_i-y_{N+1}$, $\tilde{\eta}_i=\hat{\eta}_i-\eta_i$, $\tilde{\chi}_i=\hat{\chi}_i-\chi_i$, $z_{i}=\sum_{j \in \mathcal{\bar{N}}_i}(\hat{y}_j-\hat{y}_i)$ and $\tilde{\theta}_{i}=\hat{\theta}_{i}-\theta$, for $i=1,\cdots,N$.
Then, based on \eqref{compensator2} and \eqref{neq4a}, we
can obtain:
\begin{subequations}
\label{reeq12}
\begin{align}
\dot{\tilde{\eta}}_i&= A_i \tilde{\eta}_i+B_i \tilde{y}_i,\label{reeq12a} \\
\dot{\tilde{\chi}}_i&= A^T_i\tilde{\chi}_i+E\tilde{y}_i, \label{reeq12b}\\
\dot{\tilde{\theta}}_i&=\kappa F \hat{\chi}_{i}z_{i}, \\
\dot{\tilde{y}}_i&=E^T \tilde{\eta}_{i}+a_{i,1}\tilde{y}_i+\tilde{\chi}_{i}^T
F^T \theta+\hat{\chi}_{i}^T F^T\tilde{\theta}_i+\mu z_{i}.
\end{align}\end{subequations}
Define  $\tilde{y}=\mathrm{col}(\tilde{y}_1,\ldots,\tilde{y}_N)$ and $z=%
\mathrm{col }(z_{1},\ldots,z_{N})$. Then, we have the following relation:
\begin{align}  \label{eqev}
z =-H\tilde{y}.
\end{align}
It can be shown that (\ref{reeq12}) can be written in the following compact
form:
\begin{subequations}
\label{neq5}
\begin{align}
\dot{\tilde{\eta}}=&A_d\tilde{\eta}+B_d\tilde{y} \\
\dot{\tilde{\chi}}= &A_d^T\tilde{\chi}%
+\left(I_N\otimes E\right)\tilde{y}  \label{reeq2abbbc} \\
\dot{\tilde{\theta}}=&-\kappa\left(I_N\otimes F\right)\hat{\chi}_d H\tilde{y}
\label{reeq2abbbb} \\
\dot{\tilde{y}}=&\left(a_d-\mu H\right)\tilde{y}+\left(I_N\otimes
E^T\right)\tilde{\eta}  \notag \\
&+\tilde{\chi}_{d}^T\left[\mathds{1}_N\otimes \left(F^T\theta\right)\right]+\hat{\chi}%
_{d}^T\left(I_N\otimes F^T \right)\tilde{\theta}
\end{align}\end{subequations}
where $\mathds{1}_N\in \mathds{R}^{N}$ is the vector of all ones, and%
\begin{align}
A_d&=\mathrm{blkdiag}(A_1,\ldots,A_N), B_d=\mathrm{blkdiag}(B_1,\ldots,B_N), \notag \\
\tilde{\eta}&=\mathrm{col}(\tilde{\eta}_1,\ldots,\tilde{\eta}_N),~~\tilde{%
\chi}_d=\mathrm{blkdiag} (\tilde{\chi}_1,\ldots,\tilde{\chi}_N),  \notag
\\
\tilde{\theta}&=\mathrm{col}(\tilde{\theta}_1,\ldots,\tilde{\theta}_N),~~%
\hat{\chi}_d=\mathrm{blkdiag}\left(\hat{\chi}_1,\ldots,\hat{\chi}%
_N\right),  \notag \\
\tilde{\chi}&=\mathrm{col}(\tilde{\chi}_1,\ldots,\tilde{\chi}_N),~~a_{d}=\mathrm{blkdiag}(a_{1,1},\ldots,a_{N,1}).  \notag
\end{align}
We now ready to establish our main technical lemmas.
\begin{lem}
\label{lemma4}Consider systems \eqref{leaderall}, \eqref{compensator2}, %
\eqref{neq4a} and \eqref{neq5}. For any $\hat{\eta}_i(0)$, $\hat{\chi}_i(0)$, $%
\hat{\theta}_i(0)$ and $\hat{y}(0)$, choose any $\kappa>0$ and a constant
 $\mu$ satisfying \eqref{muiequa}. Then, for $i=1,\cdots,N$, $\hat{y}_i(t)$, $\tilde{\eta}_i(t)$, $\hat{\chi}_i(t)$ and $\hat{\theta}_i(t)$ are
bounded for all $t\geq 0$ and satisfy
\begin{align}
\lim\limits_{t\rightarrow\infty}\tilde{\eta}_i(t)&=0  \label{eq101b} \\
\lim\limits_{t\rightarrow\infty}\tilde{\chi}_i(t)&=0  \label{eq101c} \\
\lim\limits_{t\rightarrow\infty}\tilde{y}_i (t)& =0  \label{eq101a} \\
\lim\limits_{t\rightarrow\infty}\hat{\chi}_{i}^T(t)F^T\tilde{\theta}_i(t)& =0.  \label{eq101d}
\end{align}
\end{lem}
\begin{proof}Since $A_i$ is stable, for the positive numbers $\bar{\gamma}_{1}\geq\gamma_{i,1}$ and $\bar{\gamma}_{2}\geq\gamma_{i,2}$, where $\gamma_{i,1}$ and $\gamma_{i,2}$ are defined in \eqref{gamma12}, there always exist  positive definite matrices $P_i\in \mathds{R}^{(2l-1)\times (2l-1)}$ and $Q_i\in \mathds{R}%
^{(2l-1)\times (2l-1)}$ satisfying the following Riccati inequalities
 \begin{subequations}\label{maeqin}
\begin{align}
P_iA_i+A_i^{T}P_i+\frac{1}{\bar{\gamma}_{1}^2}P_iB_iB_i^{T}P_i+EE^T&<0, \\
Q_iA_i^{T}+A_iQ_i+\frac{1}{\bar{\gamma}_{2}^2}Q_iEE^TQ_i+F^T F&<0.
\end{align}
\end{subequations}
%according to the bounded real lemma in \cite{boyd1994linear}.
Consider the following Lyapunov function candidate for (\ref{neq5}):
\begin{equation}\label{reeq3b}
V=\tilde{y}^{T}H\tilde{y}+\tilde{\eta}^{T}P_d\tilde{\eta}+\tilde{\chi}^{T}Q_d \tilde{\chi}+\frac{%
1}{\kappa }\tilde{\theta}^{T}\tilde{\theta},
\end{equation}
where $H$ is defined in Section~\ref{ss1}, $$P_d=\mathrm{blkdiag}(P_{1},\ldots,P_{N})~~\text{and}~~Q_d=\mathrm{blkdiag}(Q_{1},\ldots,Q_{N}).$$
 Differentiating (\ref{reeq3b}) along the trajectory of \eqref{neq5} gives
\begin{align}
\dot{V}=& 2\tilde{y}^{T}H\dot{\tilde{y}}+2\tilde{\eta}^{T}P_d \dot{\tilde{\eta}}+2\tilde{\chi}^{T}Q_d \dot{\tilde{\chi}}+\frac{2}{\kappa }\tilde{\theta}^{T}%
\dot{\tilde{\theta}}  \notag \\
=& 2\tilde{y}^{T}H\Big[\left( I_{N}\otimes E^{T}\right) \tilde{\eta}+a_{d}\tilde{y}+\tilde{\chi}_{d}^{T}\left[\mathds{1}_{N}\otimes \left(F^T\theta\right)
\right]  \notag \\
& +\hat{\chi}_{d}^{T}\left( I_{N}\otimes F^{T}\right) \tilde{\theta}-\mu
H\tilde{y}\Big]  +2\tilde{\eta}^{T}\big[P_d A_d\tilde{\eta}+P_dB_d\tilde{y}\big]  \notag \\
& +2\tilde{\chi}^{T}\big[Q_d A_d^{T}\tilde{\chi}
+Q_d\left( I_{N}\otimes E\right) \tilde{y}\big]  \notag \\
& -2\tilde{\theta}^{T}\left( I_{N}\otimes F\right) \hat{\chi}_{d}H\tilde{y}
\notag \\
=& 2\tilde{y}^{T}Ha_{d}\tilde{y}-2\mu \tilde{y}^{T}H^{2}\tilde{y}  \notag
\\
& +\tilde{\eta}^{T}\left(P_d A_d+A_d^{T}P_d\right)
\tilde{\eta}  +\tilde{\chi}^{T}\left(Q_d A_d^{T}+A_d Q_d \right)
\tilde{\chi}  \notag \\
& +2\tilde{y}^{T}H\left( I_{N}\otimes E^{T}\right) \tilde{\eta}+2\tilde{y}%
^{T}H\tilde{\chi}_{d}^{T}\left[ \mathds{1}_{N}\otimes \left(F^T\theta\right) \right]
\notag \\
& +2\tilde{\eta}^{T}P_dB_d\tilde{y}+2\tilde{\chi}%
^{T}Q_d\left( I_{N}\otimes E\right) \tilde{y}  \notag \\
& +2\tilde{y}^{T}H\hat{\chi}_{d}^{T}\left( I_{N}\otimes F^{T}\right) \tilde{\theta}-2\tilde{\theta}^{T}\left( I_{N}\otimes F\right) \hat{\chi}_{d}H\tilde{y}.
\end{align}%
It is easy to verify that
\begin{align}
\tilde{y}^{T}H\hat{\chi}_{d}^{T}\left( I_{N}\otimes F^{T}\right) \tilde{%
\theta}=\tilde{\theta}^{T}\left( I_{N}\otimes F\right) \hat{\chi}_{d}H\tilde{y}.\label{rnda}
\end{align}%
Moreover, we have%
\begin{align}
&2\tilde{y}^{T}H\tilde{\chi}_{d}^{T}\left[ \mathds{1}_{N}\otimes \left(F^T\theta\right) \right] \leq\|H\tilde{y}\|^2 \|\theta\|^2 +\|(I_N\otimes F)\tilde{\chi}\|^2,\nonumber\\
&2\tilde{y}^{T}H\left(I_N\otimes
E^T\right)\tilde{\eta}\leq \tilde{y}^{T}H^2\tilde{y}+\tilde{\eta}^T\left(I_N\otimes
E E^T\right)\tilde{\eta},\label{vvinequal1a1}
\end{align}%
where the fact \begin{align}
\left\| \tilde{\chi}_d^T\big(\mathds{1}_{N}\otimes \left(F^T\theta\right)\big) \right\|&=\left\|\big(I_{N}\otimes \theta^T\big) (I_N\otimes F)\tilde{\chi} \right\|\nonumber\\
&\leq \|\theta\| \|(I_N\otimes F)\tilde{\chi}\|\nonumber
\end{align}
is used to obtain the first inequality.

From (\ref{rnda}), and (\ref{vvinequal1a1}),  we have%, where $\pi=\max\limits_{\omega_k\leq \bar{\omega},k=1,\dots,l}\theta_i$
\begin{align}\label{vvinequal}
\dot{V}\leq&  2\tilde{y}^{T}Ha_{d}\tilde{y}-2\mu \tilde{y}^{T}H^{2}\tilde{y}+(1+\|\theta\|^2)\tilde{y}^{T}H^2\tilde{y}
\notag \\
&+\sum_{i=1}^{N}\Big[\tilde{\eta}^{T}_{i}\left( P_i A_i+A_i^{T}P_i \right)
\tilde{\eta}_{i}+2\tilde{\eta}_i^{T}P_i B_i\tilde{y}_i\notag \\
& +\tilde{\eta}_i^{T}E E^{T}\tilde{\eta}_i  \Big]+\sum_{i=1}^{N}\Big[\tilde{\chi}^{T}_i\left( Q_i A_i^{T}+A_i Q_i \right)\tilde{\chi}_i \notag \\
& +2\tilde{\chi}_i^{T}Q_i E \tilde{y}_i+\tilde{\chi}_{i}^{T}F^T F\tilde{\chi}_{i} \Big].
\end{align}
%For the positive numbers $\bar{\gamma}_{1}$ and $\bar{\gamma}_{2}$,
%From \eqref{maeqin}, we have
From \eqref{maeqin}, for $i=1,\ldots,N$, we have
\begin{align*}
&\tilde{\eta}_{i}^{T}\left( P_{i}A_{i}+A_{i}^{T}P_{i}\right) \tilde{\eta}%
_{i}+2\tilde{\eta}_{i}^{T}P_{i}B_{i}\tilde{y}_{i}+\tilde{\eta}_{i}^{T}EE^{T}\tilde{\eta}_{i} \leq \bar{\gamma}_{1}^{2}%
\tilde{y}_{i}^{T}\tilde{y}_{i}, \\
&\tilde{\chi}_{i}^{T}\left( Q_{i}A_{i}^{T}+A_{i}Q_{i}\right) \tilde{\chi}%
_{i}+2\tilde{\chi}_{i}^{T}Q_{i}E\tilde{y}_{i}+\tilde{\chi}_{i}^{T}F^T F\tilde{\chi}_{i} \leq \bar{\gamma}_{2}^{2}%
\tilde{y}_{i}^{T}\tilde{y}_{i}.
\end{align*}%
Applying the above two inequalities to \eqref{vvinequal}, we have%
\begin{align}\label{vvinequald}
\dot{V}\leq& 2\tilde{y}^{T}Ha_{d}\tilde{y}+(1+\|\theta\|^2-2\mu )\tilde{y}^{T}H^2\tilde{y}+\left({\bar{\gamma}_{1}}^2 +{\bar{\gamma}_{2}}^2\right)\tilde{y}^{T}\tilde{y}\notag\\
= &  \tilde{y}^{T}H\Big[\left(1+\|\theta\|^2-2\mu\right)I +\left({\bar{\gamma}_{1}}^2 +{\bar{\gamma}_{2}}^2\right)H^{-2}\Big] H\tilde{y}\notag\\
&+\tilde{y}^{T}H\big(a_dH^{-1}+H^{-1}a_d\big)H\tilde{y}.
\end{align}
Finally, since $\lambda_1I \leq H$, from \eqref{muiequa} and \eqref{vvinequald}, we have
\begin{align}
\dot{V}
\leq & \Big[ 1+\|\theta\|^2-2\mu  +2\bar{a}\lambda_1^{-1} +\left({\bar{\gamma}_{1}}^2 +{\bar{\gamma}_{2}}^2\right)\lambda_1^{-2}\Big]\tilde{y}^{T}H^2\tilde{y} \notag\\
\leq& 0,
\end{align}
where $\bar{a}=\max\{a_{1,1},\ldots,a_{N,1}\}$. Since $V$ is positive definite and $\dot{V}$ is negative semi-definite, $V$
is bounded, which means $\tilde{y}$, $\tilde{\eta}$, $\tilde{\chi}$ and $%
\tilde{\theta}$ are all bounded. From (\ref{neq5}), $\dot{\tilde{y}}$, $\dot{%
\tilde{\eta}}$, $\dot{\tilde{\chi}}$ and $\dot{\tilde{\theta}}$ are bounded,
which implies that $\ddot{V}$ is bounded. According to Barbalat's lemma, %in~\cite{slotine1991applied}, we have
$\lim\limits_{t\rightarrow \infty }\dot{V}(t)=0$,
which implies (\ref{eq101a}). Thus, by (\ref{eqev}), we have
$
\lim\limits_{t\rightarrow \infty }z(t)=0
$,
which together with (\ref{reeq2abbbb}) yields $\lim\limits_{t\rightarrow
\infty }\dot{\tilde{\theta}}=0$.
 As both $A_i^T$ and $A_i$ are Hurwitz matrices, systems \eqref{reeq12a} and \eqref{reeq12b} are stable systems with a bounded input $\tilde{y}_i$. The
input is bounded for all $t\geq0$ and tends to zero as $t\rightarrow\infty$. We conclude that $\tilde{\chi}_i(t)$ and $\tilde{\eta}_i$ will decay to zero as $t\rightarrow\infty
$ from the input to state stability property, for $i=1,\ldots,N$.
 To show (\ref{eq101d}), differentiating $%
\dot{\tilde{y}}$ gives,
\begin{align}
\ddot{\tilde{y}}=& \left(I_N\otimes
E^T\right) \dot{\tilde{\eta}}+a_{d}\dot{\tilde{y}}+\dot{%
\tilde{\chi}}_{d}^{T}\left[\mathds{1}_N\otimes \left(F^T\theta\right)\right]-\mu H\dot{%
\tilde{y}}  \notag  \label{reeq6b} \\
& +\dot{\hat{\chi}}_{d}^{T}\left( I_{N}\otimes F^{T}\right) \tilde{\theta}+%
\hat{\chi}_{d}^{T}\left( I_{N}\otimes F^{T}\right) \dot{\tilde{\theta}}.
\end{align}%
We have shown that $\dot{{\tilde{\eta}}}%
$, $\dot{{\tilde{y}}}$, $\dot{{\tilde{\chi}}}$, $\dot{\hat{\chi}}$, ${\tilde{\theta}}$, ${\hat{\chi}}$, and $\dot{\tilde{\theta}}$ are
all bounded. Thus, $\ddot{\tilde{y}}$ is bounded. By using Barbalat's
lemma again, we have $\lim\limits_{t%
\rightarrow \infty }\dot{\tilde{y}}(t)=0$, which together with (\ref{eq101b}), (\ref{eq101c}) and (\ref{eq101a}%
) implies $$\lim\limits_{t\rightarrow
\infty }\hat{\chi}_{d}^{T}(t)\left( I_{N}\otimes F^{T}\right) \tilde{\theta}%
(t)=0.$$
Then, we have \eqref{eq101d}.
\end{proof}
\subsection{Parameter Estimation Error Analysis}
Lemma \ref{lemma4} does not guarantee $\lim\limits_{t\rightarrow\infty}
\tilde{\theta}(t) =0$. It is possible to make $\lim\limits_{t\rightarrow%
\infty} \tilde{\theta}(t)=0$ if the signal $ F \hat{\chi}_{i} (t)$ is persistently exciting. We need the following result which is taken from Lemma 2.4 in \cite{chen2015stabilization}.
\begin{lem}
\label{lemmape} Consider a continuously differentiable function $g:[0,+\infty)\mapsto \mathds{R}^n$ and a bounded piecewise continuous
function $f:[0,+\infty)\mapsto \mathds{R}^n$, which satisfy $\lim\limits_{t\rightarrow \infty}g^T(t)f(t)=0$. Then, $\lim\limits_{t\rightarrow
\infty}g(t)=0$ holds under the following two conditions:
\begin{enumerate}[(i)]
  \item $\lim\limits_{t\rightarrow \infty}\dot{g}(t)=0$;
  \item $f(t)$ is persistently exciting.
\end{enumerate}
\end{lem}
\begin{thm}
\label{lemccci}Consider systems \eqref{leaderall}, \eqref{compensator2} and \eqref{neq4a}. For any $\hat{\eta} (0)$, $\hat{\chi} (0)$, $\hat{y}(0)$ and $\hat{\theta}(0)$, choose any $\kappa>0$ and a constant
 $\mu$ satisfying \eqref{muiequa}. Then,
 \begin{enumerate}[(i)]
 \item $F\hat{\chi}_i(t)$ is persistently exciting;
 \item
$\lim\limits_{t\rightarrow\infty} \tilde{\theta}(t)=0. $\end{enumerate}
\end{thm}

\begin{proof}(i): For the system (\ref{neq4ia}), $A_i^T\in \mathds{R}^{(2l-1)\times (2l-1)}$ is a Hurwitz matrix. It can be easily verified that $(A_i^T,E)$ is in the controllable canonical form, %controllable from Kalman criterion,
 and the input $y_{N+1}(t)$ is sufficiently rich of order $2l$, which is greater than $2l-1$. Using Theorem 2.7.2 in \cite{sastry2011adaptive}, we have $\chi_i(t)$
is persistently exciting.
Since $F\in \mathds{R}^{l\times (2l-1)}$ is a full row rank constant matrix with $\mathrm{rank}(F)=l$,
$F\chi_i(t)$ is persistently exciting according to Lemma 4.8.3 in \cite{ioannou1996robust} or Lemma 1 in \cite{narendra1987persistent}.
From (\ref{eq101c}) in Lemma~\ref{lemma4}, we have
$$
\lim\limits_{t\rightarrow\infty}F\left(\hat{\chi}_i(t)-\chi_i(t)
\right)=0,~~i=1,\cdots,N.
$$
Then, by Lemma 3.2 in \cite{wanghuang2018}, $F\hat{\chi}_i(t)$
is persistently exciting for $i=1,\cdots,N$.

(ii):
From (\ref{eq101d}) in Lemma~\ref{lemma4}, we have
$\lim\limits_{t\rightarrow\infty}\hat{\chi}_i^T(t)F ^T\tilde{\theta}_i(t)= 0$, $i=1,\cdots,N.$
Since $F\hat{\chi}_i(t)$ is persistently exciting and $
\lim\limits_{t\rightarrow\infty} \dot{\tilde{\theta}}_{i}(t)=0$ following from \eqref{reeq2abbbb}, we have $
\lim\limits_{t\rightarrow\infty} \tilde{\theta}_{i}(t)=0$ for $i=1,\dots,N$, according to Lemma \ref{lemmape}.
\end{proof}
%\begin{rem}Lemma \ref{lemccci} guarantees that $\tilde{\theta}(t)$ converges to the origin asymptotically and $t\rightarrow\infty$.\end{rem}
\subsection{State Estimation Error Analysis}
We now show that the estimation error between $\hat{v}_{i}$ in \eqref{distristab} and $v$ in \eqref{leader2} tends to zero as $t\rightarrow \infty$. Let $\tilde{v}_i=\hat{v}_i-v$. Then, we have the following
equations
\begin{align}
\dot{\tilde{v}}_i =& \left(M -L_iC \right)\tilde{v}_i+L_i\tilde{y}_i-\sum_{k=1}^{l}\big(\hat{\theta}_{i,k}E_{2k}\hat{y}_i-\theta_k E_{2k}y_{N+1}
\big)  \notag \\
=&\left(M -L_iC \right)\tilde{v}_i+L_i\tilde{y}_i-\sum_{k=1}^{l}\big(\tilde{\theta}_{i,k}E_{2k}\hat{y}_i+\theta_kE_{2k}\tilde{y}_i\big),\notag\\
&~~~~ i=1,\cdots,N.  \label{vobsev1}
\end{align}
\begin{lem}
\label{leadervob}Consider the systems \eqref{leaderall},
\eqref{compensator2}, \eqref{distrista}, \eqref{neq4a} and \eqref{vobsev1}. For any $\hat{\eta}(0)$, $\hat{\chi} (0)$, $\hat{y}(0)$, $\hat{v}_i(0)$ and $\hat{\theta}(0)$, choose $\kappa>0$ and
 a constant $\mu$ satisfying \eqref{muiequa}.
Then,
$$\lim\limits_{t\rightarrow\infty}\tilde{v}_i(t)= 0,$$ for $i=1,\cdots,N$.
\end{lem}
\begin{proof}
By Lemma \ref{lemma4} and Theorem \ref{lemccci}, both $\tilde{\theta}_i(t)$ and $\tilde{y}_i(t)$ converge to zero as $t\rightarrow\infty$, for $i=1,\ldots,N$. As $M-L_i C$
is a Hurwitz matrix, the system (\ref{vobsev1}) could be viewed as a stable
system with $\sum_{k=1}^{l}\left[\tilde{\theta}_{i,k}E_{2k}\hat{y}_i+\theta_kE_{2k}\tilde{y}_i\right]$ and $L\tilde{y}_i$ as the inputs. The
inputs are bounded for all $t\geq0$ and tend to zero as $t\rightarrow\infty$. We conclude that $\tilde{v}_i(t)$ will decay to zero as $t\rightarrow\infty
$ from the input to state stability property, for $i=1,\ldots,N$.
\end{proof}
\subsection{Output Consensus Analysis}

Motivated by the
output regulation theory in \cite{isidori1990output,huang2004nonlinear}, the regulator
equations associated with follower $i$ in \eqref{MARINEVESSEL1} and the leader
system in \eqref{leader2} are defined by (\ref{regulatorequa}).

Let $$f(v,\theta)=\mathrm{col }
\left(f_1(v,\theta),f_2(v,\theta),\ldots,f_r(v,\theta)\right),
$$
where $f_1(v,\theta)$, $f_2(v,\theta)$, $\ldots$, $f_r(v,\theta)$ are defined in \eqref{regulatorequa}.
In order to solve Problem \ref{ldlesp}, we perform the following coordinate
transformation for the follower system \eqref{MARINEVESSEL1}:
$$\bar{x}_{i}=x_i-f(v,\theta),$$
where
$\bar{x}_i=\mathrm{col}(\bar{x}_{i,1},\ldots,\bar{x}_{i,r})$ for $ i=1,\dots,N$.
Then, we can obtain the following error system
\begin{subequations}
\begin{align}
\dot{\bar{x}}_{i,s}&=\bar{x}_{i,s+1},~~~~s=1,\ldots,r-1, \\
\dot{\bar{x}}_{i,r}&=u_i-f_{r+1}(v,\theta).
\label{errorsystem}
\end{align}
\end{subequations}

\begin{thm}
Consider systems \eqref{leaderall}, \eqref{MARINEVESSEL1} and a digraph $%
\mathcal{\bar{\mathcal{G}}}$. For any $\hat{\eta}(0)$, $\hat{\chi} (0)$, $\hat{\theta}(0)$, $\hat{y}(0)$, $\hat{v%
}_i(0)$, $x_i(0)$, $i=1\dots,N$, choose $\kappa>0$ and %there exists a sufficiently large
 a constant $\mu$ satisfying \eqref{muiequa},
Problem \ref{ldlesp} is solved by the distributed adaptive observer (\ref{compensator2}) and the distributed control law %
\eqref{distrista}.
\end{thm}

\begin{proof}
Substituting \eqref{distristaa} into the error system \eqref{errorsystem}
yields the following system
\begin{align}  \label{fitxe0}
\dot{\bar{x}}_i=\Phi\bar{x}_i+g_i(t),~~~~i=1,\ldots,N,
\end{align}
where $D=\mathrm{col}\left(0_{(r-1) \times 1},1\right)$,
\begin{align} %\label{Aform}
\Phi=&\left[
\begin{array}{c|c}
0 & I_{r-1} \\ \hline
-\alpha_1 & -\alpha_2,~\ldots,~-\alpha_r \\
\end{array}
\right],\nonumber\\
%\end{align} and
%\begin{equation*}
g_i(t)=&D\sum_{s=1}^{r+1}\alpha_s\left(f_{s}(\hat{v}_i,\hat{\theta}%
_i)-f_{s}(v,\theta)\right).\nonumber
\end{align}
For $i=1,\dots,N$, we can rewrite $g_i(t)$ as%
\begin{align}  \label{fitxe}
g_i(t)=&D\sum_{s=1}^{r+1}\alpha_s\left(f_{s}(\tilde{v}_i+v,\tilde{%
\theta}_i+\theta)-f_{s}(v,\theta)\right)
\end{align}
By Lemma \ref{lemma4}, Theorem~\ref{lemccci} and Lemma \ref{leadervob}, for
any $\hat{\eta}(0)$, $\hat{\chi} (0)$, $\hat{v} (0)$, $\hat{y}(0)$ and $\hat{\theta}(0)$,
since $\kappa>0$ and $\mu $ satisfying (\ref{muiequa}), $\hat{v}(t)$ and $\hat{\theta}(t)$ are bounded for all $t\geq 0$ and satisfy
\begin{equation*}
\lim_{t\rightarrow\infty}\tilde{v}_i(t)=0,~~~~\lim_{t\rightarrow\infty}%
\tilde{\theta}_i(t)=0.
\end{equation*}
Thus, we have
$\lim\limits_{t\rightarrow\infty}g_i(t)=0$, $i=1,\dots,N.$
%\end{equation*}
As $\Phi$ is a Hurwitz matrix, the system \eqref{fitxe0} can be viewed as a stable system
with $g_i(t)$ as the input. Since this input is bounded for all $t\geq0$ and
tends to zero as $t\rightarrow\infty$, we can conclude that $\lim\limits_{t\rightarrow\infty}\bar{x}_i(t)=0$ from the input to state stability property. For $i=1,\ldots,N$, $\lim\limits_{t\rightarrow\infty}\bar{x}_i(t)=0$ implies
$$\lim\limits_{t\rightarrow\infty}\bar{x}_{i,1}(t)=0.$$
 Therefore, Problem \ref{ldlesp} is solved.
\end{proof}
\section{Simulation}\label{section5}
Consider a multi-agent system with five followers and one leader, where the underlying
communication topology is shown in Fig.~\ref{fig1} with $\lambda_1=0.5188$.
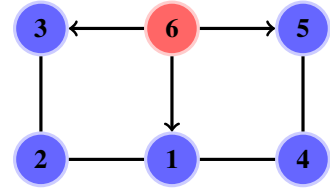
\begin{figure}[htbp]
\begin{center}
\begin{tikzpicture}[transform shape]
    \centering%
    \node (1) [circle, draw=blue!20, fill=blue!60, very thick, minimum size=7mm] {\textbf{1}};
    \node (2) [circle, left=of 1, draw=blue!20, fill=blue!60, very thick, minimum size=7mm] {\textbf{2}};
    \node (3) [circle, above=of 2, draw=blue!20, fill=blue!60, very thick, minimum size=7mm] {\textbf{3}};
    \node (4) [circle, right=of 1, draw=blue!20, fill=blue!60, very thick, minimum size=7mm] {\textbf{4}};
     \node (5) [circle, above=of 4, draw=blue!20, fill=blue!60, very thick, minimum size=7mm] {\textbf{5}};
    \node(6) [circle, above=of 1, draw=red!20, fill=red!60, very thick, minimum size=7mm]{\textbf{6}};
    \draw[ very  thick,->,  right] (6) edge (1);
    \draw[ very  thick,-,  right] (2) edge (1);
    \draw[ very  thick,-,  right] (1) edge (4);
    \draw[ very  thick,->,  right] (6) edge (3);
    \draw[ very  thick,-,  right] (2) edge (3);
    \draw[ very  thick,-,  right] (4) edge (5);
    \draw[ very  thick,->,  right] (6) edge (5);
\end{tikzpicture}
\end{center}
\caption{ Communication topology $\bar{\mathcal{G}}$}
\label{fig1}
\end{figure}

The output signal of the leader can be described by the following system%
\begin{equation}
y_{N+1}= \varphi_1\sin(\omega_{1} t+ \psi_1)+ \varphi_2\sin(\omega_{2} t +
\psi_2),\label{outsimu}
\end{equation}
where $\omega_{1}$, $\omega_{2}$ $\varphi_1$, $\varphi_2$, $\psi_1$ and $\psi_2$ are arbitrary unknown real numbers with $\omega_{1},\omega_{2}\in(0,1.5]$. Here $\theta_1=\omega_1^2+\omega_2^2$ and $\theta_2=\omega_1^2\omega_2^2$ according to \eqref{char1}. Thus, $\pi=45.8789$ from \eqref{pi1}.

The dynamics of the followers are given below:
\begin{align}\label{numerica1}
\dot{x}_{i,1} &= x_{i,2},~~\dot{x}_{i,2} =u_i,\notag \\
y_i&=x_{i,1},~~i=1,\ldots,5.
\end{align}

Five agents choose the vectors $a_1=\mathrm{col}\left(2.5,2.49,1.49\right)$, $a_2=\mathrm{col}(2,2,1.5)$, $a_3=\mathrm{col}(2,2.5,1.2)$, $a_4=\mathrm{col}(2,3,2)$ and $a_5=\mathrm{col}(1.5,2,1)$ such that all the roots of the polynomial equation \eqref{aplychar1} have negative real parts, respectively.
Then, we can calculate $\bar{a}=2.5$, $\gamma_{1,1}=2.7039$, $\gamma_{2,1}=2.4526$, $\gamma_{3,1}=2.3007$, $\gamma_{4,1}=2.6303$, $\gamma_{5,1}=2.1196$, $\gamma_{1,2}=0.8037$, $\gamma_{2,2}=1.2680$, $\gamma_{3,2}= 0.8333$, $\gamma_{4,2}=0.9132$ and $\gamma_{5,2}= 1.3219$ from \eqref{gamma12}. Thus, we can choose $\gamma_{1}=2.7039$ and $\gamma_{2}=  1.3219$.
Hence, we have $\mu> 45.0857$ from \eqref{muiequa}.
 The distributed adaptive observer \eqref{compensator2} can be
designed with $\mu =56$ and $\kappa=500$. The distributed control law can be designed in the form of \eqref{distrista}, where
$\alpha_1=6$ and $\alpha_2=11$ are chosen such that all the roots of the polynomial equation \eqref{alphacontrol} have negative real parts, $E_{2}=\mathrm{col}\left(0,1,0,0\right)$ and $%
E_{4}=\mathrm{col}\left(0,0,0,1\right)$, and $L_i=%
\mathrm{col}\left(12,54,108,81\right)$ such that $M-L_iC$ is a Hurwitz matrix, where  $M$ and $C$ are defined in \eqref{MC}, $i=1,\ldots,5$.  %by the formula

Simulation is conducted
with the following initial conditions: $\hat{\chi}_{i}(0)=0$, $\hat{\eta}_{i}(0)=0$, $\hat{y}_i(0)=0$,  $x_{i}(0)=0$, and $\hat{\theta}_i(0)=0$, $i=1,\dots,5$. Fig.~\ref{fig2} shows the tracking
errors $e_{i}(t)=y_{i}(t)-y_{N+1}(t)$ for $i=1,\dots,5$, which all converge to the origin as time $t\rightarrow\infty$. The unknown amplitudes and phases of $y_{N+1}$ in \eqref{outsimu} are $\varphi_1=5$, $\varphi_2=2$, $\psi_1=0$ and $\psi_2=0$.
\begin{figure}[ht]
\centering
\epsfig{figure=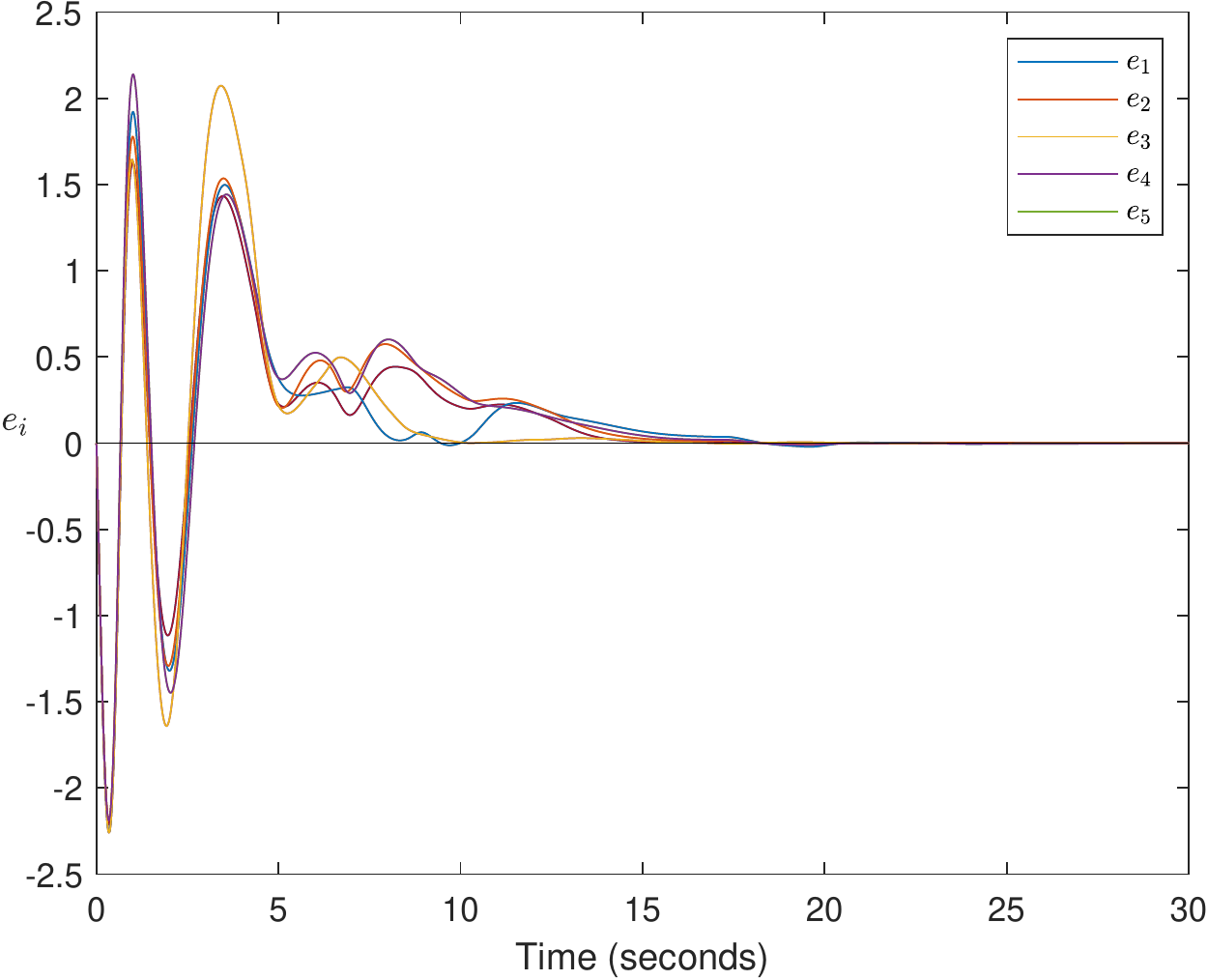,height=2.7in}  % \caption{Tracking Error}\label{fig2i}
\caption{Trajectories of $e_i(t)$, $i=1,\ldots,5$}
\label{fig2}
\end{figure}

Fig.~\ref{figthe} shows the trajectories of $\|\hat{\theta}_i(t)-\theta\|$ of each agent. It shows that the proposed distributed observer can estimate the actual unknown parameters $\theta=\mathrm{col}(3.25,2.25)$ asymptotically.
\begin{figure}[ht]
\centering
\epsfig{figure=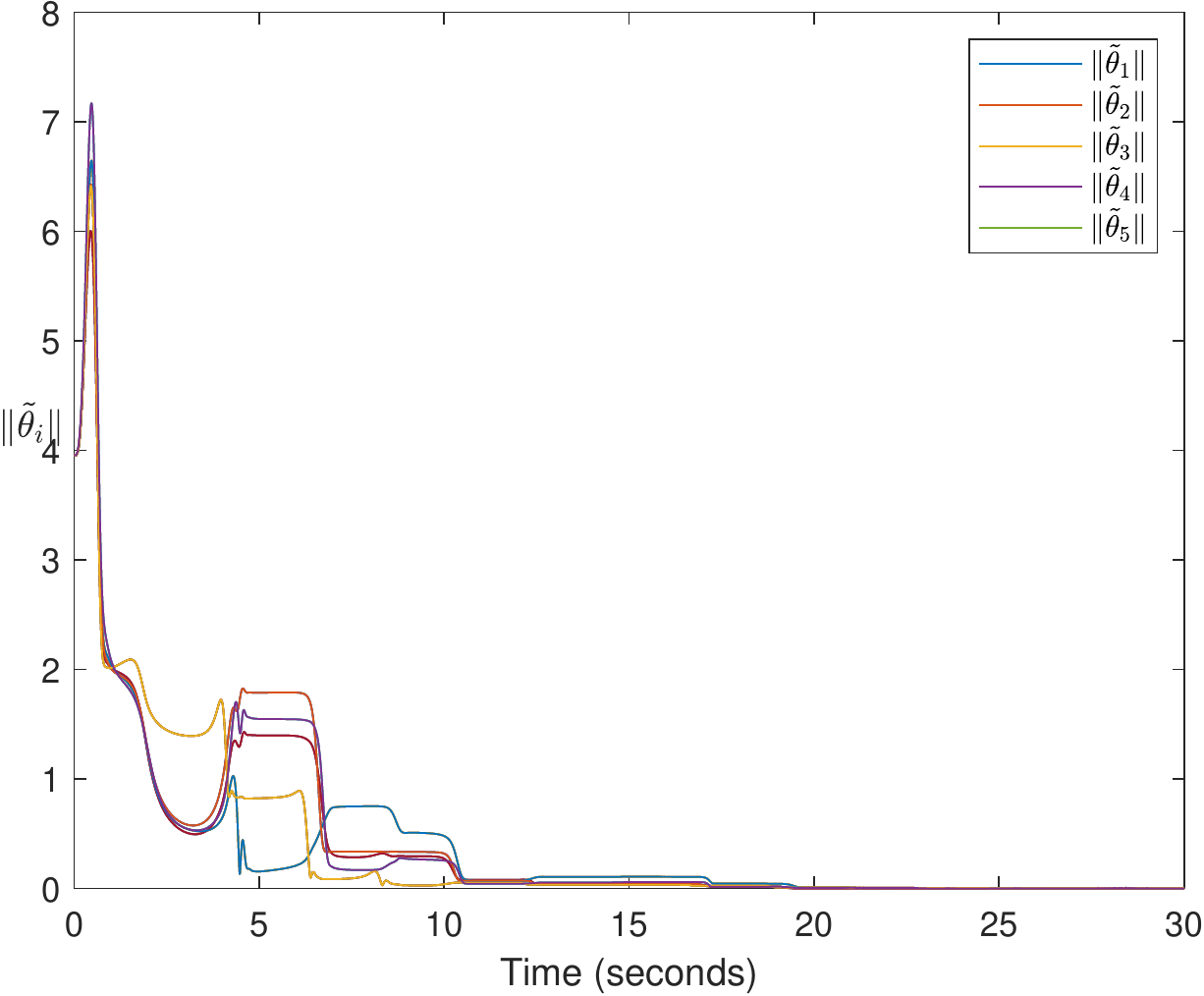,height=2.7in}  % \caption{Tracking Error}\label{fig2i}
\caption{Trajectories of $\|\hat{\theta}_i(t)-\theta\|$, $i = 1, \ldots, 5.$}
\label{figthe}
\end{figure}
\begin{figure}[ht]
\centering
\epsfig{figure=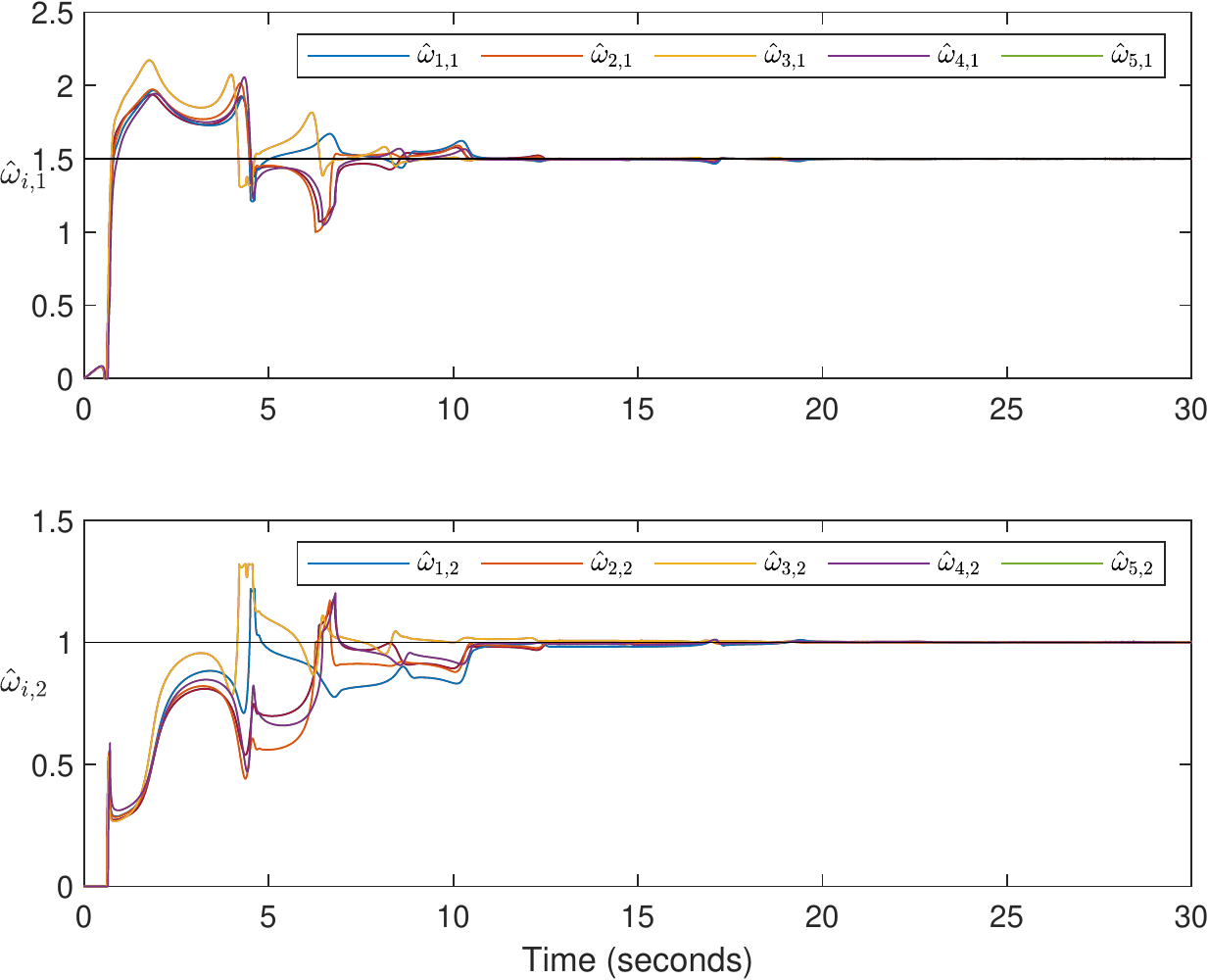,height=2.7in}  % \caption{Tracking Error}\label{fig2i}
\caption{Trajectories of $\hat{\protect\omega}_i(t)$ with $\omega=\mathrm{col}(1.5,1)$, $i = 1, \ldots, 5.$}
\label{fig5}
\end{figure}
The distributed adaptive observer
directly provides the estimate of $\theta=\mathrm{col}\left(\theta_{1},\theta_{2}\right)$
which are related to the estimates of $\omega_1$ and $\omega_2$ through the following equations
 \begin{equation}\label{oemgasolu}
\hat{\omega}_{1i,2i,}=\sqrt{\frac{\hat{\theta}_{1i}\pm \sqrt{\hat{\theta%
}_{1i}^2-4\hat{\theta}_{2i}}}{2}},~~~~i=1,\ldots,5.
\end{equation}
Fig.~\ref{fig5} shows the trajectories of $\hat{\omega}_i(t)$, $i=1,\dots,5$,
which, as expected, all converge to the actual unknown value of $\omega=\mathrm{col}(1.5,1)$. As noted in Fig.~\ref{fig5}, there are some time intervals where $\hat{\omega}_i(t)=0$, for $i=1,\dots,5$. This is because constraints on $\theta$ through the reparameterization in (\ref{char1}) are not taken into account by the observer design in (\ref{compensator2c}). Hence, the calculation of $\hat{\omega}_i(t)$ using \eqref{oemgasolu} may produce complex numbers. When this happens, $\hat{\omega}_i(t)$ is simply considered as $0$.

\section{Conclusions}
\label{section6}The leader-following consensus problem of a multi-agent system subject to an uncertain leader
system has attracted great interest, and yet approaches for simultaneous parameter estimation and consensus are relatively few. This paper proposed a framework for designing distributed adaptive observers and distributed observer based controllers that ensure asymptotic consensus under unknown parameters of the leader. It was shown that the parameter estimation errors of all agents converge to zero asymptotically when the signal of the leader is sufficiently rich. Furthermore, computations of the bounds involved in the design of the distributed adaptive observer and controller were provided explicitly.

%\bibliographystyle{ieeetr}
%\bibliography{myref}

\begin{thebibliography}{36}

\bibitem{jadbabaie2003coordination}
A.~Jadbabaie, J.~Lin, and A.~S. Morse, ``Coordination of groups of mobile
  autonomous agents using nearest neighbor rules,'' {\em IEEE Trans. Autom. Control}, vol.~48, no.~6, pp.~988--1001, 2003.

\bibitem{olfati2007consensus}
R.~Olfati-Saber, J.~A. Fax, and R.~M. Murray, ``Consensus and cooperation in
  networked multi-agent systems,'' {\em Proceedings of the IEEE}, vol.~95,
  no.~1, pp.~215--233, 2007.

\bibitem{liu2016finite}
X.~Liu, J.~Lam, W.~Yu, and G.~Chen, ``Finite-time consensus of multi-agent
  systems with a switching protocol,'' {\em IEEE Trans. Neural Netw. Learn. Syst}, vol.~27, no.~4, pp.~853--862, 2016.

\bibitem{cao2010decentralized}
Y.~Cao, W.~Ren, and Z.~Meng, ``Decentralized finite-time sliding mode
  estimators and their applications in decentralized finite-time formation
  tracking,'' {\em Systems \& Control Letters}, vol.~59, no.~9, pp.~522--529,
  2010.

\bibitem{cao2011formation}
M.~Cao, C.~Yu, and B.~D. Anderson, ``Formation control using range-only
  measurements,'' {\em Automatica}, vol.~47, no.~4, pp.~776--781, 2011.


\bibitem{wangsong2017}Y.~Wang, Y.~Song, and W,~Ren, "Distributed adaptive finite-time approach for formation–containment control of networked nonlinear systems under directed topology," {\em IEEE Trans. Neural Netw. Learn. Syst}, vol.29, no.7, pp.~3164--3175, 2017.

\bibitem{feng2019finite}
Z.~Feng, G.~Hu, and C.~G. Cassandras, ``Finite-time distributed convex
  optimization for continuous-time multi-agent systems with disturbance
  rejection,'' {\em IEEE Trans. Control Netw. Syst}, vol.~7, no.~2, pp.~686--698, 2019.

\bibitem{lin2018distributed}
P.~Lin, W.~Ren, C.~Yang, and W.~Gui, ``Distributed optimization with nonconvex
  velocity constraints, nonuniform position constraints, and nonuniform
  stepsizes,'' {\em IEEE Trans. Autom. Control}, vol.~64, no.~6,
  pp.~2575--2582, 2018.

\bibitem{cai2016leaderEL}
H.~Cai and J.~Huang, ``The leader-following consensus for multiple uncertain
  Euler-Lagrange systems with an adaptive distributed observer,'' {\em IEEE Trans. Autom. Control}, vol.~61, no.~10, pp.~3152--3157, 2016.

\bibitem{xiao2017adaptive}
F.~Xiao and T.~Chen, ``Adaptive consensus in leader-following networks of
  heterogeneous linear systems,'' {\em IEEE Trans. Control Netw. Syst}, vol.~5, no.~3, pp.~1169--1176, 2017.


\bibitem{chen2019feedforward}
Z.~Chen, ``Feedforward design for output synchronization of nonlinear
  heterogeneous systems with output communication,'' {\em Automatica},
  vol.~104, pp.~126--133, 2019.

\bibitem{su2011cooperative}
Y.~Su and J.~Huang, ``Cooperative output regulation of linear multi-agent
  systems,'' {\em IEEE Trans. Autom. Control}, vol.~57, no.~4,
  pp.~1062--1066, 2011.

\bibitem{yu2010distributed}
W.~Yu, G.~Chen, and M.~Cao, ``Distributed leader--follower flocking control for
  multi-agent dynamical systems with time-varying velocities,'' {\em Systems \&
  Control Letters}, vol.~59, no.~9, pp.~543--552, 2010.

\bibitem{modares2016optimal}
H.~Modares, S.~P. Nageshrao, G.~A.~D. Lopes, R.~Babu{\v{s}}ka, and F.~L. Lewis,
  ``Optimal model-free output synchronization of heterogeneous systems using
  off-policy reinforcement learning,'' {\em Automatica}, vol.~71, pp.~334--341,
  2016.

\bibitem{wu2017adaptive}
Y.~Wu, R.~Lu, P.~Shi, H.~Su, and Z.~G.~Wu, ``Adaptive output synchronization of
  heterogeneous network with an uncertain leader,'' {\em Automatica}, vol.~76,
  pp.~183--192, 2017.

\bibitem{lu2019leader}
M.~Lu and L.~Liu, ``Leader-following consensus of multiple uncertain
  Euler-Lagrange systems with unknown dynamic leader,'' {\em IEEE Trans. Autom. Control}, vol.~64, no.~10, pp.~4167--4173, 2019.

\bibitem{zhao2015distributed}
Y.~Zhao, Z.~Duan, G.~Wen, and G.~Chen, ``Distributed finite-time tracking for a
  multi-agent system under a leader with bounded unknown acceleration,'' {\em
  Systems \& Control Letters}, vol.~81, pp.~8--13, 2015.

\bibitem{chung2009cooperative}
J.~J.~E. Slotine, and W.~ Li, {\em Applied nonlinear control}, vol.~199, no.~1, \newblock Englewood Cliffs, NJ: Prentice hall, 1991.


\bibitem{rimon1992new}
E.~Rimon and K.~S.~Narendra, ``A new adaptive estimator for linear systems,'' {\em
  IEEE Trans. Autom. Control}, vol.~37, no.~3, pp.~410--412, 1992.

\bibitem{narendra1974stableI}
K.~S. Narendra and P.~Kudva, ``Stable adaptive schemes for system
  identification and control-part i,'' {\em IEEE Trans. Syst. Man. Cybern.}, no.~6, pp.~542--551, 1974.

\bibitem{narendra1974stableII}
K.~S. Narendra and P.~Kudva, ``Stable adaptive schemes for system
  identification and control-part ii,'' {\em IEEE Trans. Syst. Man. Cybern.}, no.~6, pp.~552--560, 1974.

\bibitem{ioannou1996robust}
P.~A. Ioannou and J.~Sun, {\em Robust adaptive control}, vol.~1.
\newblock PTR Prentice-Hall Upper Saddle River, NJ, 1996.

\bibitem{adetola2014adaptive}
V.~Adetola, M.~Guay, and D.~Lehrer, ``Adaptive estimation for a class of
  nonlinearly parameterized dynamical systems,'' {\em IEEE Trans. Autom. Control}, vol.~59, no.~10, pp.~2818--2824, 2014.

\bibitem{jiang2019parameter}
T.~Jiang, D.~Xu, T.~Chen, and A.~Sheng, ``Parameter estimation of discrete-time
  sinusoidal signals: A nonlinear control approach,'' {\em Automatica},
  vol.~109, p.~108510, 2019.

\bibitem{hsuliu1999}L.~Hsu, R.~Ortega, and G.~Damm. "A globally convergent frequency estimator."  {\em IEEE Trans. Autom. Control}, vol.~44, no.4, pp.~ 698-713, 1999.

\bibitem{li2019kernel}
P.~Li, F.~Boem, G.~Pin, and T.~Parisini, ``Kernel-based simultaneous
  parameter-state estimation for continuous-time systems,'' {\em IEEE Trans. Autom. Control}, vol.~65, no.~7, pp.~3053--3059, 2020.

\bibitem{zhu2015parameter}
J.~Zhu, X.~Lin, R.~S. Blum, and Y.~Gu, ``Parameter estimation from quantized
  observations in multiplicative noise environments,'' {\em IEEE Trans. Signal Processing}, vol.~63, no.~15, pp.~4037--4050, 2015.

\bibitem{carnevale2013semi}
D.~Carnevale and A.~Astolfi, ``Semi-global multi-frequency estimation in the
  presence of deadzone and saturation,'' {\em IEEE Trans. Autom. Control}, vol.~59, no.~7, pp.~1913--1918, 2013.

\bibitem{wanghuang2018}
S.~Wang and J.~Huang, ``Adaptive leader-following consensus for multiple
  Euler-Lagrange systems with an uncertain leader system,'' {\em IEEE Trans. Neural Netw. Learn. Syst}, vol.~30, no.~7,
  pp.~2188--2196, 2019.

%\bibitem{wanghuang2020}S.~Wang and J.~Huang, ``Adaptive distributed observer for an uncertain leader with an unknown output over directed acyclic graphs," {\em Int. J. Control}, Doi.org/10.1080/00207179.2020.1766117, 2020.

\bibitem{godsil2013algebraic}
C.~Godsil and G.~F. Royle, {\em Algebraic graph theory}, vol.~207.
\newblock Springer Science \& Business Media, 2013.

\bibitem{marino2002global}
R.~Marino and P.~Tomei, ``Global estimation of n unknown frequencies,'' {\em
  IEEE Trans. Autom. Control}, vol.~47, no.~8, pp.~1324--1328,
  2002.

\bibitem{sastry2011adaptive}
S.~Sastry and M.~Bodson, {\em Adaptive control: stability, convergence and
  robustness}.
\newblock Courier Corporation, 2011.

\bibitem{chen2015stabilization}
Z.~Chen and J.~Huang, ``Stabilization and regulation of nonlinear systems,''
  {\em Cham, Switzerland: Springer}, 2015.

\bibitem{narendra1987persistent}
K.~S. Narendra and A.~M. Annaswamy, ``Persistent excitation in adaptive
  systems,'' {\em Int. J. Control}, vol.~45, no.~1,
  pp.~127--160, 1987.

\bibitem{isidori1990output}
A.~Isidori and C.~I. Byrnes, ``Output regulation of nonlinear systems,'' {\em
  IEEE Trans. Autom. Control}, vol.~35, no.~2, pp.~131--140, 1990.

\bibitem{huang2004nonlinear}
J.~Huang, {\em Nonlinear output regulation: theory and applications}, vol.~8.
\newblock Siam, 2004.

\end{thebibliography}

\end{document}